\documentclass[11pt,a4paper,leqno]{article}
\usepackage[latin1]{inputenc}
\usepackage[english]{babel}

\usepackage{amsfonts}
\usepackage{amsmath}
\usepackage{amssymb}
\usepackage{amsthm} % for proof environment
\usepackage{bm} % for bm
\usepackage[round]{natbib}
\usepackage{graphicx}
\usepackage{caption}
\usepackage{subcaption}
\usepackage[parfill]{parskip}
\usepackage{epstopdf}
\usepackage[space]{grffile}
\usepackage{appendix}
\usepackage{booktabs}
\usepackage{blkarray}
\usepackage{verbatim}
\usepackage{siunitx}
\usepackage[colorlinks=true,citecolor=blue]{hyperref}

\newcommand{\dd}{{\mathrm d}}

\newcommand{\e}{{\rm e}}
\newcommand{\E}{{\mathbb E}}

\newcommand{\Q}{{\mathbb Q}}

\newcommand{\R}{{\mathbb R}}
\newcommand{\N}{{\mathbb N}}

\newcommand{\Fcal}{{\mathcal F}}
\newcommand{\Gcal}{{\mathcal G}}

\newcommand{\Ocal}{{\mathcal O}}

\makeatletter
\def\thm@space@setup{%
  \thm@preskip=\parskip \thm@postskip=0pt
}
\makeatother

\newtheorem{proposition}{Proposition}[section]
\newtheorem{lemma}[proposition]{Lemma}

\newtheorem{corollary}[proposition]{Corollary}
\newtheorem{remark}[proposition]{Remark}

\title{Asian Option Pricing with Orthogonal Polynomials\footnote{I thank Damien Ackerer, Damir Filipovi\'c, participants at the 9th International Workshop on Applied Probability in Budapest, and two anonymous referees for helpful comments. The research leading to these results has received funding from the European Research Council under the European Union's Seventh Framework Programme (FP/2007-2013) / ERC Grant Agreement n.~307465-POLYTE.}}
\author{Sander Willems\footnote{EPFL and Swiss Finance Institute, 1015 Lausanne, Switzerland. Email: sander.willems@epfl.ch}}
\date{September 14, 2018\\ forthcoming in \emph{Quantitative Finance}}
\numberwithin{equation}{section}
\begin{document}

\maketitle

\begin{abstract}
In this paper we derive a series expansion for the price of a continuously sampled arithmetic Asian option in the Black-Scholes setting. The expansion is based on polynomials that are orthogonal with respect to the log-normal distribution. All terms in the series are fully explicit and no numerical integration nor any special functions are involved. We provide sufficient conditions to guarantee convergence of the series. The moment indeterminacy of the log-normal distribution introduces an asymptotic bias in the series, however we show numerically that the bias can safely be ignored in practice.

\medskip


\end{abstract}

\section{Introduction}
An Asian option is a derivative contract with payoff contingent on the average value of the underlying asset over a certain time period. Valuation of these contracts is not straightforward because of the path-dependent nature of the payoff. Even in the standard \cite{black1973pricing} setting the distribution of the (arithmetic) average stock price is not known. In this paper we derive a series expansion for the Asian option price in the Black-Scholes setting using polynomials that are orthogonal with respect to the log-normal distribution. The terms in the series are fully explicit since all the moments of the average price are known. We prove that the series does not diverge by showing that the tails of the average price distribution are dominated by the tails of a log-normal distribution. As a consequence of the well known moment indeterminacy of the log-normal distribution (see e.g., \cite{heyde2010property}), it is not theoretically guaranteed that the series converges to the true price. We show numerically that this asymptotic bias is small for standard parameterizations and the real approximation challenge lies in controlling the error coming from truncating the series after a finite number of terms.

There exists a vast literature on the problem of Asian option pricing. We give a brief overview which is by no means exhaustive. One approach is to approximate the unknown distribution of the average price with a more tractable one. \cite{turnbull1991quick}, \cite{levy1992pricing}, \cite{ritchken1993valuation}, \cite{li2016pricing} use an Edgeworth expansion around a log-normal reference distribution to approximate the distribution of the arithmetic average of the geometric Brownian motion. \cite{ju2002pricing} and \cite{sun2013quasi} use a Taylor series approach to approximate the unknown average distribution from a log-normal. \cite{milevsky1998asian} use a moment matched inverse gamma distribution as approximation. Their choice is motivated by the fact that the infinite horizon average stock price has an inverse gamma distribution. More recently, \cite{aprahamian2015pricing} use a moment matched compound gamma distribution. Although these type of approximations lead to analytic option price formulas, their main drawback is the lack of reliable error estimates. A second strand of the literature focuses on Monte-Carlo and PDE methods. \cite{kemna1990pricing} propose to use the continuously sampled geometric option price as a control variate and show that it leads to a significant variance reduction. \cite{fu1999pricing} argue that this is a biased control variate, but interestingly the bias approximately offsets the bias coming from discretely computing the continuous average in the simulation. \cite{lapeyre2001competitive} perform a numerical comparison of different Monte-Carlo schemes. \cite{rogers1995value}, \cite{zvan1996robust}, \cite{vecer2001new,vecer2002unified}, \cite{marcozzi2003valuation} solve the pricing PDE numerically. Another approach is to derive bounds on the Asian option price, see e.g. \cite{curran1994valuing}, \cite{rogers1995value}, \cite{thompson2002fast}, and \cite{vanmaele2006bounds}. Finally, there are several papers that derive exact representations of the Asian option price. \cite{yor1992some} expresses the option price as a triple integral, to be evaluated numerically. \cite{geman1993bessel} derive the Laplace transform of the option price. Numerical inversion of this Laplace transform is however a delicate task, see e.g. \cite{eydeland1995domino}, \cite{fu1999pricing}, \cite{shaw2002pricing}. \cite{carverhill1990valuing} relate the density of the discrete arithmetic average to an iterative convolution of densities, which is approximated numerically through the Fast Fourier Transform algorithm. Later extensions and improvements of the convolution approach include \cite{benhamou2002fast}, \cite{fusai2008pricing}, \cite{vcerny2011improved}, and \cite{fusai2011pricing}.
\cite{dufresne2000laguerre} derives a series representation using Laguerre orthogonal polynomials.  \cite{linetsky2004spectral} derives a series representation using spectral expansions involving Whittaker functions.

The approach taken in this paper is closely related to \cite{dufresne2000laguerre} in the sense that both are based on orthogonal polynomial expansions. The Laguerre series expansion can be shown to diverge when directly expanding the density of the average price, which is related to the fact that the tails of the average price distribution are heavier than those of the Gamma distribution. As a workaround, \cite{dufresne2000laguerre} proposes to work with the reciprocal of the average, for which the Laguerre series does converge. The main downside of this approach is that the moments of the reciprocal average are not available in closed form and need to be calculated through numerical integration, which introduces a high computational cost and additional numerical errors. \cite{asmussen2016orthonormal} use a different workaround and expand an exponentially tilted transformation of the density of a sum of log-normal random variables using a Laguerre series. They show that the exponential tilting transformation guarantees the expansion to converge. However, a similar problem as in \cite{dufresne2000laguerre} arises: the moments of the exponentially tilted density are not available in closed form and have to be computed numerically. In contrast, our approach is fully explicit and does not involve any numerical integration, which makes it very fast.

Truncating our series after only one term is equivalent to pricing the option under the assumption that the average price is log-normally distributed. The remaining terms in the series can therefore be thought of as corrections to the log-normal distribution. This has a very similar flavour to approaches using an Edgeworth expansion around the log-normal distribution (cfr. \cite{jarrow1982approximate} and \cite{turnbull1991quick}). The key difference with our approach is that the Edgeworth expansion can easily diverge because it lacks a proper theoretical framework. In contrast, the series we present in this paper is guaranteed to converge, possibly with a small asymptotic bias. A thorough study of the approximation error reveals that the asymptotic bias is positively related to the volatility of the stock price process and the option expiry. We use the integration by parts formula from Malliavin calculus to derive an upper bound on the approximation error.

The remaining of this paper is structured as follows. Section \ref{section:distribution_average} casts the problem and derives useful properties about the distribution of the arithmetic average. Section \ref{section:pol_exp} describes the density expansion used to approximate the option price. In Section \ref{section:error_bounds} we investigate the approximation error. Section \ref{section:numerical_examples} illustrates the method with numerical examples. Section \ref{section:conclusion} concludes. All proofs can be found in Appendix \ref{appendix:proofs}.

\section{The distribution of the arithmetic average} \label{section:distribution_average}
We fix a stochastic basis $(\Omega,\Fcal,(\Fcal_t)_{t\ge 0},\Q)$ satisfying the usual conditions and let $\Q$ be the risk-neutral probability measure. We consider the Black-Scholes setup where the underlying stock price $S_t$ follows a geometric Brownian motion:
\begin{align*}
\dd S_t=r S_t\,\dd t+\sigma S_t\,\dd B_t,
\end{align*}
where $r\in\R$ is the short-rate, $\sigma>0$ the volatility of the asset, and $B_t$ a standard Brownian motion. For notational convenience we assume $S_0=1$, which is without loss of generality. We define the average price process as
\begin{align*}
A_t=\frac{1}{t}\int^t_0S_u\,\dd u,\quad t> 0.
\end{align*}
The price at time 0 of an Asian option with continuous arithmetic averaging, strike $K>0$, and expiry $T>0$ is given by
\begin{align*}
\pi&=\e^{-rT}\E\left[\left(A_T-K\right)^+\right].
\end{align*}
The option price can not be computed explicitly since we do not know the distribution of $A_T$. However, we can derive useful results about its distribution.

We start by computing all the moments of $A_T$. Using the time-reversal property of a Brownian motion, we have the following identity in law (cfr.\ \cite{dufresne1990distribution}, \cite{carmona1997distribution}, \cite{donati2001certain}, \cite{linetsky2004spectral}):
\begin{lemma}
The random variable $TA_T$ has the same distribution as the solution at time $T$ of the following SDE
\begin{align}
\dd X_t=(rX_t+1)\,\dd t+\sigma X_t\,\dd B_t,\quad X_0=0.
\label{eq:garch_diffusion}
\end{align}
\label{lemma:identity_law}
\end{lemma}
Since the SDE in \eqref{eq:garch_diffusion} defines a polynomial diffusion (see e.g., \cite{filipovic2016polynomial}), we can compute all the moments of its solution at a given time in closed form. By the identity in law of Lemma \ref{lemma:identity_law} we therefore also have all of the moments of $A_T$ in closed form:
\begin{proposition}\label{prop:moments}
If we denote by $H_n(x)=(1,x,\ldots,x^n)^\top$, $n\in\N$, then we have
\begin{align*}
\E\left[H_n(A_T)\right]=\e^{G_nT}H_n(0), 
\end{align*}
where $G_n\in\R^{(n+1)\times(n+1)}$ is the following lower bidiagonal matrix
\begin{align}
G_n=
\begin{pmatrix}
0&\\
\frac{1}{T}&r&\\
&\ddots&\ddots&\\
&&\frac{n}{T}&(nr+\frac{1}{2}n(n-1)\sigma^2)
\end{pmatrix}.
\label{eq:generator_matrix}
\end{align}
\end{proposition}

Given that the matrix exponential is a standard built-in function in most scientific computing packages, the above moment formula is very easy to implement. There also exist efficient numerical methods to directly compute the action of the matrix exponential, see e.g. \cite{al2011computing} and \cite{caliari2014comparison}.
An equivalent, but more cumbersome to implement, representation of the moments can be found in \cite{geman1993bessel}. 

The following proposition shows that $A_T$ admits a smooth density function $g(x)$ whose tails are dominated by the tails of a log-normal density function:
\begin{proposition}\label{prop:upper_tail}
\leavevmode
\begin{enumerate}
\item The random variable $A_T$ admits an infinitely differentiable density function $g(x)$.
\item The density function $g(x)$ has the following asymptotic properties:
\begin{align*}
g(x)=
\begin{cases}
\Ocal \left( \exp\left\{-\dfrac{3}{2}\dfrac{\log(x)^2}{\sigma^2T}\right\} \right)&\text{for}\quad x\to 0,\\[10pt]
\Ocal \left( \exp\left\{-\dfrac{1}{2}\dfrac{\log(x)^2}{\sigma^2T}\right\} \right)& \text{for}\quad x\to \infty.
\end{cases}
%
%\Ocal\left(\exp\left\{-\frac{\log(x)^2}{2\nu^2}\right\} \right),\quad \text{for}\quad x\to\infty \quad \text{and}\quad x\to 0.
\end{align*}
\end{enumerate}
\end{proposition}

\section{Polynomial expansion}\label{section:pol_exp}
Following a similar structure as \cite{ackerer2017jacobi} and \cite{ackerer2017option}, we use in this section the density expansion approach described by \cite{filipovic2013density} to approximate the Asian option price. 
Define the weighted Hilbert space $L^2_w$ as the set of measurable functions $f$ on $\R$ with finite $L^2_w$-norm defined as
\begin{align}
\lVert f \rVert ^2_w= \int_0^\infty f(x)^2w(x)\,\dd x,\quad w(x)=\frac{1}{\sqrt{2\pi}\nu x}\exp\left\{-\frac{(\log(x)-\mu)^2}{2\nu^2}\right\},\label{eq:hilbert_space}
\end{align}
for some constants $\mu\in\R$,  $\nu>0$. The weight function $w$ is the density function of a log-normal distribution. The corresponding scalar product between two functions $f,h\in L^2_w$ is defined as
\begin{align*}
\langle f,h\rangle_w=\int_0^\infty f(x)h(x)w(x)\,\dd x.
\end{align*} 

Since the measures associated with the densities $g$ and $w$ are equivalent, we can define the likelihood ratio function $\ell$ such that
\begin{align*}
g(x)=\ell(x)w(x), \quad x\in (0,\infty).
\end{align*}
Using Proposition \ref{prop:upper_tail} we now have the following result:
\begin{proposition}\label{prop:likelihood_ratio}
If $\nu^2> \frac{1}{2} \sigma^2 T$, then $\ell\in L^2_w$, i.e.
\[
\int_{0}^\infty \left(\frac{g(x)}{w(x)}\right)^2 w(x)\,\dd x<\infty.
\]
\end{proposition}

 Denote by $\mathrm{Pol}(\R)$ the set of polynomials on $\R$ and by $\mathrm{Pol}_N(\R)$ the subset of polynomials on $\R$ of degree at most $N\in\N$. Since the log-normal distribution has finite moments of any degree, we have $\mathrm{Pol}_N(\R)\subset L^2_w$ for all $N\in\N$. 
Let $b_0,b_1,\ldots, b_N$ form an orthonormal polynomial basis for $\mathrm{Pol}_N(\R)$. Such a basis can, for example, be constructed numerically from the monomial basis using a Cholesky decomposition. Indeed, define the Hankel moment matrix $M=(M_{ij})_{0\le i,j\le N}$ as
\begin{align}
M_{ij}=\langle x^i,x^j\rangle_w =\e^{\mu (i+j)+\frac{1}{2}(i+j)^2\nu^2},\quad  i,j=0,\ldots,N,\label{eq:hankel}
\end{align}
which is positive definite by construction. If we denote by $M=LL^\top$ the unique Cholesky decomposition of $M$, then 
\begin{align*}
(b_0(x),\ldots ,b_N(x))^\top
= L^{-1} H_N(x),
\end{align*}
forms an orthonormal polynomial basis for $\text{Pol}_N(\R)$. Alternative approaches to build an orthonormal basis are, for example, the three-term recurrence relation (see Lemma \ref{lemma:orthonormal_polynomials} for details) or the analytical expressions for the orthonormal polynomials derived in Theorem 1.1 of \cite{asmussen2016orthonormal}.

\begin{remark}
The matrix $M$ defined in \eqref{eq:hankel} can in practice be non-positive definite due to rounding errors. This problem becomes increasingly important for large $N$ and/or large $\nu$ because the elements in $M$ grow very fast. Similarly, the moments of $A_T$ can also grow very large, which causes rounding errors in finite precision arithmetic. In Appendix \ref{appendix:scaling} we describe a convenient scaling technique that solves these problems in many cases.
\end{remark}

Define the discounted payoff function $F(x)=\e^{-r T}(x-K)^+$. Since $F(x)\le \e^{-rT} x$ for all $x\ge 0$, we immediately have that $F\in L^2_w$. Denote by $\overline{\mathrm{Pol}(\R)}$ the closure of $\mathrm{Pol}(\R)$ in $L^2_w$. We define the projected option price $\bar{\pi}=\E[\bar{F}(A_T)]$, where $\bar{F}$ denotes the orthogonal projection of $F$ onto $\overline{\mathrm{Pol}(\R)}$ in $L^2_w$. Elementary functional analysis gives
\begin{align}
\bar{\pi}&=\langle \bar{F},\ell\rangle_w=\sum_{n\ge 0}f_n\ell_n,\label{eq:asymptotic_series}
\end{align}
where we define the likelihood coefficients $\ell_n=\langle \ell,b_n\rangle_w$ and payoff coefficients $f_n=\langle F,b_n\rangle_w$. Truncating the series in \eqref{eq:asymptotic_series} after a finite number of terms finally gives the following approximation for the Asian option price:
\begin{align}
\pi^{(N)}=\sum_{n=0}^N f_n\ell_n,\quad N\in\N. \label{eq:truncated_series}
\end{align}

The likelihood coefficients are available in closed form using the moments of $A_T$ in Proposition \ref{prop:moments}:
\begin{equation*}
(\ell_0,\ldots,\ell_N)^\top=L^{-1}\e^{G_N T}H_N(0).
\end{equation*}
The payoff coefficients can also be derived explicitly, as shown in the following proposition.
\begin{proposition}\label{prop:payoff_coeff}
Let $\Phi$ be the standard normal cumulative distribution function. The payoff coefficients $f_0,\ldots,f_N$ are given by
\begin{align*}
(f_0,\ldots,f_N)^\top=e^{-rT}L^{-1}(\tilde{f}_0,\ldots,\tilde{f}_N)^\top,
\end{align*}
with 
\begin{align}
&\tilde{f}_n=\e^{\mu (n+1) +\frac{1}{2}(n+1)^2\nu^2}\Phi(d_{n+1})-K\e^{\mu n +\frac{1}{2}n^2\nu^2}\Phi(d_n),\quad n=0,\ldots,N,\nonumber\\
& d_n=\frac{\mu+\nu^2 n-\log(K)}{\nu},\quad n=0,\ldots,N+1.\label{eq:payoff_coeff_d}
\end{align}
\end{proposition}

Equivalently, we could also derive the approximation \eqref{eq:truncated_series} by projecting $\ell$, instead of $F$, on the set of polynomials. This leads to the interpretation of \eqref{eq:truncated_series} as the option price one obtains when approximating the true density $g(x)$ by
\begin{align}
g^{(N)}(x)=w(x)\sum_{n=0}^N\ell_n b_n(x).
\label{eq:density_approx}
\end{align}
The function $g^{(N)}(x)$ integrates to one by construction
\begin{align*}
\int_0^\infty g^{(N)}(x)\,\dd x=\sum_{n=0}^N\ell_n \langle b_n,b_0=1\rangle_w=\ell_0=1,
\end{align*}
where the last equality follows from the fact that $g(x)$ integrates to one.
However, it is not a true probability density function since it is not guaranteed to be non-negative.

\section{Approximation error}\label{section:error_bounds}
The error introduced by the the approximation in \eqref{eq:truncated_series} can be decomposed as
\begin{align*}
\pi-\pi^{(N)}=\left(\pi-\bar{\pi} \right)+\left(\bar{\pi}-\pi^{(N)}\right).
\end{align*}
The second term is guaranteed to converge to zero as $N\to\infty$. In order for the first term to vanish, we need $F\in\overline{\mathrm{Pol}(\R)}$ and/or $\ell\in\overline{\mathrm{Pol}(\R)}$. It is well known (see e.g., \cite{heyde2010property}) that the log-normal distribution is not determined by its moments. As a consequence, the set of polynomials does not lie dense in $L^2_w$: $\overline{\mathrm{Pol}(\R)}\subsetneq L^2_w$. Hence, the fact that $F,\ell\in L^2_w$ is not sufficient to guarantee that the first term in the error decomposition is zero.  One of the goals of this paper is to quantify the importance of the first error term. In this section we therefore investigate the error associated with projecting $F$ and $\ell$ on the set of polynomials.

The $L^2_w$-distances of $F$ and $\ell$ to their respective orthogonal projections on $\text{Pol}_N(\R)$ are given by
\begin{align*}
\epsilon_N^F&:=\left\lVert F- \sum_{n=0}^Nb_n f_n \right \rVert_w^2=\lVert F\rVert^2_w- \sum_{n=0}^N f_n^2,\\
 \epsilon_N^\ell&:=\left\lVert\ell- \sum_{n=0}^Nb_n \ell_n \right \rVert_w^2=\lVert \ell\rVert^2_w- \sum_{n=0}^N \ell_n^2.
\end{align*}
The $L^2_w$-norm of the payoff function $F$ can be derived explicitly following very similar steps as in the proof of Proposition \ref{prop:payoff_coeff}:
\begin{align*}
\lVert F\rVert^2_w=\e^{-2rT}\left(\e^{2\mu+2\nu^2}\Phi(d_2)-2K\e^{\mu+0.5\nu^2}\Phi(d_1)+K^2\Phi(d_0)\right),
\end{align*}
where $d_0,d_1$, and $d_2$ are defined in \eqref{eq:payoff_coeff_d}. Hence, we can explicitly evaluate $\epsilon^F_N$.

The computation of $\epsilon_N^\ell$ is more difficult since $\lVert \ell \rVert^2_w$ depends on the unknown density $g(x)$. The following lemma uses the integration by parts formula from Malliavin calculus to derive a representation of $g(x)$ in terms of an expectation, which can be evaluated by Monte-Carlo simulation:
\begin{lemma}
For any $x\in\R$ we have
\begin{align}
g(x)=\E\left[ \left(1_{\{A_T\ge x\}} -c(x)\right) \frac{2}{\sigma^2}\left( \frac{S_T-S_0}{T A_T^2}+\frac{\sigma^2-r}{A_T}\right)\right],\label{eq:density_malliavin}
\end{align}
where $c$ is any deterministic finite-valued function.
\label{lemma:malliavin}
\end{lemma}

\begin{remark}
The purpose of the function $c$ is to guarantee that the simulated $g(x)$ actually goes to zero for $x\to 0$. Indeed, if we set $c(x)\equiv 0$, then $g(0)$ can be different from zero due to the Monte-Carlo error, which can lead to numerical problems when evaluating $\ell(x)$. This can be avoided by, for example, using the indicator function $c(x)=1_{x\le \xi}$, for some $\xi>0$.
\end{remark}

As a direct consequence of \eqref{eq:density_malliavin} we get the following expression for the $L^2_w$-norm of the likelihood ratio.
\begin{corollary}\label{corollary:likelihood_ratio_MC}
The $L^2_w$-norm of $\ell$ is given by
\begin{align*}
\lVert \ell\rVert^2_w=
\E&\left[ \frac{(1_{\{A_T\ge \tilde{A}_T\}} -c(\tilde{A}_T)) \frac{2}{\sigma^2}\left( \frac{S_T-S_0}{T A_T^2}+\frac{\sigma^2-r}{A_T}\right)}{w(\tilde{A}_T)}\right],
\end{align*}
where the random variable $\tilde{A}_T$ is independent from all other random variables and has the same distribution as $A_T$.
\end{corollary}

This allows us to get an estimate for $\epsilon_N^\ell$ by simulating the random vector $(S_T,A_T,\tilde{A}_T)$. In Appendix \ref{appendix:control_variate} we describe how to use the known density function of the geometric average as a powerful control variate to significantly reduce the variance in the Monte-Carlo simulation.

Using the Cauchy-Schwarz inequality we also have the following upper bound on the approximation error in terms of the projection errors $\epsilon_N^F$ and $\epsilon_K^\ell$.
\begin{proposition}\label{prop:error_bounds}
For any $N\ge 0$ we have
\begin{align}
|\pi-\pi^{(N)}|\le\sqrt{\epsilon_N^F\epsilon_N^\ell}\label{eq:error_bound}
\end{align}
\end{proposition}

This upper bound will therefore be small if $F$ and/or $\ell$ are well approximated by a polynomial in $L^2_w$. Computing the upper bound involves a Monte-Carlo simulation to compute $\epsilon_N^\ell$, which makes it impractical to use as a decision rule for $N$. This bound should be seen as a more conservative estimate of the approximation error compared to direct simulation of the option price.

\section{Numerical examples}\label{section:numerical_examples}
In this section we compute Asian option prices using the series expansion in \eqref{eq:truncated_series}. The orthonormal basis is constructed using the scaling technique described in Appendix \ref{appendix:scaling}. We set  $\nu^2=\frac{1}{2}\sigma^2T+10^{-4}$ so that Proposition \ref{prop:likelihood_ratio} is satisfied and choose $\mu$ so that the first moment of $w$ matches the first moment of $A_T$. As a consequence, we always have 
\[\ell_1=\int_0^\infty b_1(x)g(x)\,\dd x=\langle b_0,b_1\rangle_w=0.\]

\begin{remark}
Choosing $\mu$ and $\nu$ so that the first two moments of $A_T$ are matched is typically not possible due to the restriction $\nu^2>\frac{1}{2}\sigma^2 T$ in Proposition \ref{prop:likelihood_ratio}. A similar problem arises in the Jacobi stochastic volatility model of \cite{ackerer2017jacobi}, where options are priced using a polynomial expansion with a normal density as weight function. \cite{ackerer2017option} address this problem by using a mixture of two normal densities as weight function. Specifying $w$ as a mixture of normal densities would not work  in our setting since in this case $\ell\notin L^2_w$ .\footnote{
Instead of approximating the distribution of $A_T$, it is tempting to approximate the distribution of $\log(A_T)$ and rewrite the discounted payoff function accordingly. In this case, one can show that specifying $w$ as a normal density function gives a series approximation that converges to the true price. The catch is that we do not know the moments of $\log(A_T)$, only those of $A_T$, and hence the terms in the series can  not be computed explicitly.} Instead, we can use a mixture of two log-normal densities:
\begin{align*}
w(x)=c w_1(x)+(1-c)w_2(x),
\end{align*}
where $c\in[0,1]$ is the mixture weight, and $w_1$ and $w_2$ are log-normal density functions with mean parameters $\mu_1,\mu_2\in\R$, and volatility parameters $\nu_1,\nu_2>0$, respectively. In order for Proposition \ref{prop:likelihood_ratio} to apply, it suffices to choose $\nu_1^2>\frac{1}{2}\sigma^2 T$. The remaining parameters can then be chosen freely and used for higher order moment matching. 
\end{remark}

We consider a set of seven parameterizations that has been used as a set of test cases in, among others, \cite{eydeland1995domino}, \cite{fu1999pricing}, \cite{dufresne2000laguerre}, and \cite{linetsky2004spectral}. The first columns of Table \ref{table:benchmark_cases} contain the parameter values of the seven cases. The cases are ordered in increasing size of $\tau=\sigma^2T$. Remark that $S_0\neq 1$ for all cases, however we normalize the initial stock price to one and scale the strike and option price accordingly. The columns LNS10, LNS15, LNS20 (Log-Normal Series) contain the option price approximations using \eqref{eq:truncated_series} for $N=10,15,20$, respectively. The columns LS (Laguerre Series) and EE (Eigenvalue Expansion) correspond to the series expansions of \cite{dufresne2000laguerre} and \cite{linetsky2004spectral}, respectively. The column VEC shows the prices produced by the PDE method of \cite{vecer2001new,vecer2002unified} using a grid with 400 space points and 200 time points.\footnote{This grid choice corresponds to the one used in \cite{vecer2001new}. By significantly increasing the number of space points in the grid, the PDE method can achieve the same accuracy as \cite{linetsky2004spectral}. However, doing so makes the method very slow.}  The last column contains the 95\% confidence intervals of a Monte-Carlo simulation using the geometric Asian option price as a control variate, cfr.\ \cite{kemna1990pricing}. We simulate $2\times 10^5$ price trajectories with a time step of $10^{-3}$ and approximate the continuous average with a discrete average. 

\begin{center}
[Table \ref{table:benchmark_cases} around here]
\end{center}

For the first three cases we find virtually identical prices as \cite{linetsky2004spectral}, which is one of the most accurate benchmarks available in the literature. Remarkably, our method does not face any problems with the very low volatility in case 1. Many other existing method have serious difficulty with this parameterization. Indeed, the series of \cite{dufresne2000laguerre} does not even come close to the true price, while \cite{linetsky2004spectral} requires 400 terms to obtain an accurate result. The price of \cite{vecer2001new,vecer2002unified} is close to the true price, but still outside of the 95\% Monte-Carlo confidence interval. Methods based on numerical inversion of the Laplace transform of \cite{geman1993bessel} also struggle with low volatility because they involve numerical integration of highly oscillating integrands (see e.g., \cite{fu1999pricing}). When using exact arithmetic for case 1, our series with 20+1 terms agrees with the 400 term series of \cite{linetsky2004spectral} to eight decimal places. When using double precision arithmetic, which was used for all numerical results in this section, the price agrees to four decimal places due to rounding errors. For cases 4 to 7, the LNS prices are slightly different from the EE benchmark. However, they are still very close and with the exception of case 4, they are all within the 95\% confidence interval of the Monte-Carlo simulation. 

\begin{center}
[Figure \ref{fig:benchmark_cases} around here]
\end{center}

Figure \ref{fig:benchmark_cases} plots the LNS price approximations for $N$ ranging from 0 to 20, together with the Monte-Carlo price and the corresponding 95\% confidence intervals as a benchmark.\footnote{Figure \ref{fig:benchmark_cases} and \ref{fig:density_approx} only show cases 1, 3, 5, and 7. The plots for the other cases look very similar and are available upon request.} We observe that the series converges very fast in all cases. In fact, truncating the series at $N=10$ would give almost identical results. In theory, $N$ can be chosen arbitrarily large, however in finite precision arithmetic it is inevitable that rounding errors start playing a role at some point. Remark that the prices for $N=0$ and $N=1$ are identical, which is a consequence of the fact that the auxiliary distribution matches the first moment of $A_T$. Figure \ref{fig:density_approx} plots the simulated true density $g(x)$ and the approximating densities $g^{(0)}(x)$, $g^{(4)}(x)$, and $g^{(20)}(x)$ defined in \eqref{eq:density_approx}. The true density was simulated using \eqref{eq:density_control_variate} in Appendix \ref{appendix:control_variate}, which is an extension of \eqref{eq:density_malliavin} using the density of the geometric average as a powerful control variate. Note that $g^{(0)}(x)=w(x)$, since $b_0(x)= \ell_0 =1$. We can see that the approximating densities approach the true density as we include more terms in the expansion. In Figure \ref{fig:density_approx_1} and \ref{fig:density_approx_3} the approximation $g^{(20)}(x)$ is virtually indistinguishable from the true density. However, in Figure \ref{fig:density_approx_5} and \ref{fig:density_approx_7} there remains a noticeable difference between $g(x)$ and $g^{(20)}(x)$. This is consistent with the pricing errors we observed earlier in Table \ref{table:benchmark_cases}.

\begin{center}
[Figure \ref{fig:density_approx} around here]
\end{center}

The above results indicate that for $\tau$ not too high, the LNS provides a very accurate approximation of the option price. This is not entirely surprising since $\tau$ determines the volatility parameter of the auxiliary log-normal density $w$, and hence how fast the tails of $w$ go to zero. Loosely speaking, when $\tau$ is small, projecting the payoff or likelihood ratio function on the set of polynomials in $L^2_w$ is almost like approximating a continuous function on a compact interval by polynomials. However, when $\tau$ becomes larger, the tails of $w$ decay slowly and it becomes increasingly difficult to approximate a function with a polynomial in $L^2_w$. In other words, for larger values of $\tau$, the moment indeterminacy of the log-normal distribution starts playing a more prominent role. A natural question is therefore whether this poses a problem for option pricing purposes. In Figure \ref{fig:range_vol_MC} we fix $T=1$ and plot for a range of values for $\sigma$ the squared relative error 
\[
SRE=\left(\frac{\hat{\pi}^{MC}-\pi^{(20)}}{\pi^{(20)}}\right)^2,
\]
where $\hat{\pi}^{MC}$ denotes the Monte-Carlo price estimate. The error starts to becomes noticeable around $\sigma=80\%$, where $\sqrt{SRE}\approx 0.5\%$. For higher values of $\sigma$ the error increases sharply. In Figure \ref{fig:range_vol_upper} we plot a more conservative estimate of the squared relative error using the upper bound in \eqref{eq:error_bound}. This plot shows that the upper bound is only significantly different from zero for $\sigma$ larger than approximately $70\%$. Figure \ref{fig:high_vol} gives a more detailed insight in the extreme case of $\sigma=100\%$. Although the LNS series converges relatively fast, it is clear from Figure \ref{fig:high_vol_MC} that it does not converge to the true price. The reason, as already mentioned before, is that the payoff and likelihood ratio functions are not accurately approximated by polynomials in the $L^2_w$-norm, as indicated by the projection errors in Figure \ref{fig:high_vol_payoffProjection} and \ref{fig:high_vol_likelihoodProjection}. We would obtain similar results by keeping $\sigma$ fixed and varying the maturity $T$, since the crucial parameter for the asymptotic pricing error is $\tau=\sigma^2T$. As a rule of thumb, we suggest to use the LNS method when $\tau\le 0.5$.

\begin{center}
[Figure \ref{fig:range_vol} and \ref{fig:high_vol} around here]
\end{center}

The main advantage of the method proposed in this paper is the ease of of its implementation and the computation speed. All terms in the series are fully explicit and involve only simple linear algebra operations. Table \ref{table:computation_times} shows the computation times of the LNS with $N\in\{10,25,20\}$, as well as the computation times of the benchmark methods.\footnote{For the LS, all symbolic calculations related to the moments of the reciprocal of the average have been pre-computed using Matlab's Symbolic Math Toolbox. We use 15+1 terms in the series, a higher number of terms leads to severe rounding problems in double precision arithmetic. For the EE, the integral representation (16) in \cite{linetsky2004spectral} has been implemented instead of the series representation (15). The implementation of the series representation involves partial derivatives of the Whittaker function with respect to its indices. These derivatives are not available in Matlab's Symbolic Math Toolbox and numerical finite difference approximations did not give accurate results. All numerical integrations are performed using Matlab's built-in function \texttt{integral}. For case 1, the numerical integration in the EE did not finish in a reasonable amount of time. For the VEC method, we use Prof.\ Jan Vecer's  Matlab implementation, which can be downloaded at \url{http://www.stat.columbia.edu/~vecer/asiancontinuous.m}. All computations are performed on a desktop computer with an Intel Xeon 3.50GHz CPU.}
The LNS computation times are all in the order of miliseconds. Although the EE is very accurate, it comes at the cost of long computation times (in the order of several seconds) caused by the expensive evaluations of the Whittaker function. The LS does not require calls to special functions, however the method is slowed down by the numerical integration involved in computing the moments of the reciprocal of the average. The implementation of both the EE and LS require the use of software that can handle symbolic mathematics, in contrast to the implementation of the LNS. The VEC method is the fastest among the benchmarks considered in this paper, but still an order of magnitude slower than the LNS.

\begin{center}
[Table \ref{table:computation_times} around here]
\end{center}

\section{Conclusion}\label{section:conclusion}
We have presented a series expansion for the continuously sampled arithmetic Asian option using polynomials that are orthogonal with respect to the log-normal distribution. The terms in the series are fully explicit and do not require any numerical integration or special functions, which makes the method very fast. We have shown that the series does not diverge if the volatility of the auxiliary log-normal distribution is sufficiently high. However, the series is not guaranteed to converge to the true price. We have investigated this asymptotic bias numerically and found that its magnitude is related to the size of $\tau=\sigma^2 T$.

There are several extensions to our method. First of all, we can handle discretely monitored Asian options using exactly the same setup, but replacing the moments of the continuous average with those of the discrete average. The latter are easily computed using iterated expectations. Secondly, we only look at fixed-strike Asian options in this paper. Since the process $\left(S_t,\int_0^t S_u\,\dd u\right)$ is jointly a polynomial diffusion, we can compute all of its mixed moments. Our method can then be extended to price floating-strike Asian options by using a bivariate log-normal as auxiliary distribution. Finally, we can define the auxiliary density $w$ as a mixture of log-normal densities, as studied in \cite{ackerer2017option}. Using a mixture allows to match higher order moments, which can lead to a faster convergence of the approximating series.

%%%%%%%%%%%%%%%%%%%%%%%%%%%%%%%%%%%%%%%%%%%%%%%%%%%%%%%%%%
%% Bibliography
%%%%%%%%%%%%%%%%%%%%%%%%%%%%%%%%%%%%%%%%%%%%%%%%%%%%%%%%%%
\clearpage
\bibliography{references}

\begin{thebibliography}{}

\bibitem[\protect\citeauthoryear{Ackerer and Filipovi{\'c}}{Ackerer and
  Filipovi{\'c}}{2017}]{ackerer2017option}
Ackerer, D. and D.~Filipovi{\'c} (2017).
\newblock Option pricing with orthogonal polynomial expansions.
\newblock {\em Working paper\/}.

\bibitem[\protect\citeauthoryear{Ackerer, Filipovi{\'c}, and Pulido}{Ackerer
  et~al.}{2017}]{ackerer2017jacobi}
Ackerer, D., D.~Filipovi{\'c}, and S.~Pulido (2017).
\newblock The {J}acobi stochastic volatility model.
\newblock {\em Finance and Stochastics\/}, 1--34.

\bibitem[\protect\citeauthoryear{Al-Mohy and Higham}{Al-Mohy and
  Higham}{2011}]{al2011computing}
Al-Mohy, A.~H. and N.~J. Higham (2011).
\newblock Computing the action of the matrix exponential, with an application
  to exponential integrators.
\newblock {\em SIAM Journal on Scientific Computing\/}~{\em 33\/}(2), 488--511.

\bibitem[\protect\citeauthoryear{Aprahamian and Maddah}{Aprahamian and
  Maddah}{2015}]{aprahamian2015pricing}
Aprahamian, H. and B.~Maddah (2015).
\newblock Pricing {A}sian options via compound gamma and orthogonal
  polynomials.
\newblock {\em Applied Mathematics and Computation\/}~{\em 264}, 21--43.

\bibitem[\protect\citeauthoryear{Asmussen, Goffard, and Laub}{Asmussen
  et~al.}{2016}]{asmussen2016orthonormal}
Asmussen, S., P.-O. Goffard, and P.~J. Laub (2016).
\newblock Orthonormal polynomial expansions and lognormal sum densities.
\newblock {\em arXiv preprint arXiv:1601.01763\/}.

\bibitem[\protect\citeauthoryear{Bally}{Bally}{2003}]{bally2003elementary}
Bally, V. (2003).
\newblock {\em An elementary introduction to Malliavin calculus}.
\newblock Ph.\ D. thesis, INRIA.

\bibitem[\protect\citeauthoryear{Benhamou}{Benhamou}{2002}]{benhamou2002fast}
Benhamou, E. (2002).
\newblock Fast {F}ourier transform for discrete {A}sian options.
\newblock {\em Journal of Computational Finance\/}~{\em 6\/}(1), 49--68.

\bibitem[\protect\citeauthoryear{Bj{\"o}rck and Pereyra}{Bj{\"o}rck and
  Pereyra}{1970}]{bjorck1970solution}
Bj{\"o}rck, k. and V.~Pereyra (1970).
\newblock Solution of {V}andermonde systems of equations.
\newblock {\em Mathematics of Computation\/}~{\em 24\/}(112), 893--903.

\bibitem[\protect\citeauthoryear{Black and Scholes}{Black and
  Scholes}{1973}]{black1973pricing}
Black, F. and M.~Scholes (1973).
\newblock The pricing of options and corporate liabilities.
\newblock {\em Journal of Political Economy\/}~{\em 81\/}(3), 637--654.

\bibitem[\protect\citeauthoryear{Caliari, Kandolf, Ostermann, and
  Rainer}{Caliari et~al.}{2014}]{caliari2014comparison}
Caliari, M., P.~Kandolf, A.~Ostermann, and S.~Rainer (2014).
\newblock Comparison of software for computing the action of the matrix
  exponential.
\newblock {\em BIT Numerical Mathematics\/}~{\em 54\/}(1), 113--128.

\bibitem[\protect\citeauthoryear{Carmona, Petit, and Yor}{Carmona
  et~al.}{1997}]{carmona1997distribution}
Carmona, P., F.~Petit, and M.~Yor (1997).
\newblock On the distribution and asymptotic results for exponential
  functionals of {L}{\'e}vy processes.
\newblock {\em Exponential functionals and principal values related to
  {B}rownian motion\/}, 73--121.

\bibitem[\protect\citeauthoryear{Carverhill and Clewlow}{Carverhill and
  Clewlow}{1990}]{carverhill1990valuing}
Carverhill, A. and L.~Clewlow (1990).
\newblock Valuing average rate ({A}sian) options.
\newblock {\em Risk\/}~{\em 3\/}(4), 25--29.

\bibitem[\protect\citeauthoryear{{\v{C}}ern{\`y} and Kyriakou}{{\v{C}}ern{\`y}
  and Kyriakou}{2011}]{vcerny2011improved}
{\v{C}}ern{\`y}, A. and I.~Kyriakou (2011).
\newblock An improved convolution algorithm for discretely sampled {A}sian
  options.
\newblock {\em Quantitative Finance\/}~{\em 11\/}(3), 381--389.

\bibitem[\protect\citeauthoryear{Curran}{Curran}{1994}]{curran1994valuing}
Curran, M. (1994).
\newblock Valuing {A}sian and portfolio options by conditioning on the
  geometric mean price.
\newblock {\em Management science\/}~{\em 40\/}(12), 1705--1711.

\bibitem[\protect\citeauthoryear{Donati-Martin, Ghomrasni, and
  Yor}{Donati-Martin et~al.}{2001}]{donati2001certain}
Donati-Martin, C., R.~Ghomrasni, and M.~Yor (2001).
\newblock On certain {M}arkov processes attached to exponential functionals of
  {B}rownian motion; applications to {A}sian options.
\newblock {\em Revista Matematica Iberoamericana\/}~{\em 17\/}(1), 179.

\bibitem[\protect\citeauthoryear{Dufresne}{Dufresne}{1990}]{dufresne1990distribution}
Dufresne, D. (1990).
\newblock The distribution of a perpetuity, with applications to risk theory
  and pension funding.
\newblock {\em Scandinavian Actuarial Journal\/}~{\em 1990\/}(1), 39--79.

\bibitem[\protect\citeauthoryear{Dufresne}{Dufresne}{2000}]{dufresne2000laguerre}
Dufresne, D. (2000).
\newblock Laguerre series for {A}sian and other options.
\newblock {\em Mathematical Finance\/}~{\em 10\/}(4), 407--428.

\bibitem[\protect\citeauthoryear{Eydeland and Geman}{Eydeland and
  Geman}{1995}]{eydeland1995domino}
Eydeland, A. and H.~Geman (1995).
\newblock Domino effect: Inverting the {L}aplace transform.
\newblock {\em Risk\/}~{\em 8\/}(4), 65--67.

\bibitem[\protect\citeauthoryear{Filipovi{\'c} and Larsson}{Filipovi{\'c} and
  Larsson}{2016}]{filipovic2016polynomial}
Filipovi{\'c}, D. and M.~Larsson (2016).
\newblock Polynomial diffusions and applications in finance.
\newblock {\em Finance and Stochastics\/}~{\em 20\/}(4), 931--972.

\bibitem[\protect\citeauthoryear{Filipovi{\'c}, Mayerhofer, and
  Schneider}{Filipovi{\'c} et~al.}{2013}]{filipovic2013density}
Filipovi{\'c}, D., E.~Mayerhofer, and P.~Schneider (2013).
\newblock Density approximations for multivariate affine jump-diffusion
  processes.
\newblock {\em Journal of Econometrics\/}~{\em 176\/}(2), 93--111.

\bibitem[\protect\citeauthoryear{Fourni{\'e}, Lasry, Lebuchoux, Lions, and
  Touzi}{Fourni{\'e} et~al.}{1999}]{fournie1999applications}
Fourni{\'e}, E., J.-M. Lasry, J.~Lebuchoux, P.-L. Lions, and N.~Touzi (1999).
\newblock Applications of {M}alliavin calculus to {M}onte {C}arlo methods in
  finance.
\newblock {\em Finance and Stochastics\/}~{\em 3\/}(4), 391--412.

\bibitem[\protect\citeauthoryear{Fu, Madan, and Wang}{Fu
  et~al.}{1999}]{fu1999pricing}
Fu, M.~C., D.~B. Madan, and T.~Wang (1999).
\newblock Pricing continuous {A}sian options: a comparison of {M}onte {C}arlo
  and {L}aplace transform inversion methods.
\newblock {\em Journal of Computational Finance\/}~{\em 2\/}(2), 49--74.

\bibitem[\protect\citeauthoryear{Fusai, Marazzina, and Marena}{Fusai
  et~al.}{2011}]{fusai2011pricing}
Fusai, G., D.~Marazzina, and M.~Marena (2011).
\newblock Pricing discretely monitored asian options by maturity randomization.
\newblock {\em SIAM Journal on Financial Mathematics\/}~{\em 2\/}(1), 383--403.

\bibitem[\protect\citeauthoryear{Fusai and Meucci}{Fusai and
  Meucci}{2008}]{fusai2008pricing}
Fusai, G. and A.~Meucci (2008).
\newblock Pricing discretely monitored {A}sian options under {L}{\'e}vy
  processes.
\newblock {\em Journal of Banking \& Finance\/}~{\em 32\/}(10), 2076--2088.

\bibitem[\protect\citeauthoryear{Gautschi}{Gautschi}{2004}]{gautschi2004orthogonal}
Gautschi, W. (2004).
\newblock {\em Orthogonal Polynomials: Computation and Approximation}.
\newblock Oxford University Press.

\bibitem[\protect\citeauthoryear{Geman and Yor}{Geman and
  Yor}{1993}]{geman1993bessel}
Geman, H. and M.~Yor (1993).
\newblock Bessel processes, {A}sian options, and perpetuities.
\newblock {\em Mathematical Finance\/}~{\em 3\/}(4), 349--375.

\bibitem[\protect\citeauthoryear{Heyde}{Heyde}{1963}]{heyde2010property}
Heyde, C. (1963).
\newblock On a property of the lognormal distribution.
\newblock In {\em Selected Works of CC Heyde}. Springer Science \& Business
  Media.

\bibitem[\protect\citeauthoryear{Jarrow and Rudd}{Jarrow and
  Rudd}{1982}]{jarrow1982approximate}
Jarrow, R. and A.~Rudd (1982).
\newblock Approximate option valuation for arbitrary stochastic processes.
\newblock {\em Journal of Financial Economics\/}~{\em 10\/}(3), 347--369.

\bibitem[\protect\citeauthoryear{Ju}{Ju}{2002}]{ju2002pricing}
Ju, N. (2002).
\newblock Pricing {A}sian and basket options via {T}aylor expansion.
\newblock {\em Journal of Computational Finance\/}~{\em 5\/}(3), 79--103.

\bibitem[\protect\citeauthoryear{Kemna and Vorst}{Kemna and
  Vorst}{1990}]{kemna1990pricing}
Kemna, A.~G. and A.~Vorst (1990).
\newblock A pricing method for options based on average asset values.
\newblock {\em Journal of Banking \& Finance\/}~{\em 14\/}(1), 113--129.

\bibitem[\protect\citeauthoryear{Lapeyre, Temam, et~al.}{Lapeyre
  et~al.}{2001}]{lapeyre2001competitive}
Lapeyre, B., E.~Temam, et~al. (2001).
\newblock Competitive {M}onte {C}arlo methods for the pricing of {A}sian
  options.
\newblock {\em Journal of Computational Finance\/}~{\em 5\/}(1), 39--58.

\bibitem[\protect\citeauthoryear{Levy}{Levy}{1992}]{levy1992pricing}
Levy, E. (1992).
\newblock Pricing {E}uropean average rate currency options.
\newblock {\em Journal of International Money and Finance\/}~{\em 11\/}(5),
  474--491.

\bibitem[\protect\citeauthoryear{Li and Chen}{Li and
  Chen}{2016}]{li2016pricing}
Li, W. and S.~Chen (2016).
\newblock Pricing and hedging of arithmetic {A}sian options via the {E}dgeworth
  series expansion approach.
\newblock {\em The Journal of Finance and Data Science\/}~{\em 2\/}(1), 1--25.

\bibitem[\protect\citeauthoryear{Linetsky}{Linetsky}{2004}]{linetsky2004spectral}
Linetsky, V. (2004).
\newblock Spectral expansions for {A}sian (average price) options.
\newblock {\em Operations Research\/}~{\em 52\/}(6), 856--867.

\bibitem[\protect\citeauthoryear{Marcozzi}{Marcozzi}{2003}]{marcozzi2003valuation}
Marcozzi, M.~D. (2003).
\newblock On the valuation of {A}sian options by variational methods.
\newblock {\em SIAM Journal on Scientific Computing\/}~{\em 24\/}(4),
  1124--1140.

\bibitem[\protect\citeauthoryear{Milevsky and Posner}{Milevsky and
  Posner}{1998}]{milevsky1998asian}
Milevsky, M.~A. and S.~E. Posner (1998).
\newblock Asian options, the sum of lognormals, and the reciprocal gamma
  distribution.
\newblock {\em Journal of Financial and Quantitative Analysis\/}~{\em 33\/}(3),
  409--422.

\bibitem[\protect\citeauthoryear{Nualart}{Nualart}{2006}]{nualart2006malliavin}
Nualart, D. (2006).
\newblock {\em The {M}alliavin Calculus and Related Topics}.
\newblock Springer.

\bibitem[\protect\citeauthoryear{Ritchken, Sankarasubramanian, and
  Vijh}{Ritchken et~al.}{1993}]{ritchken1993valuation}
Ritchken, P., L.~Sankarasubramanian, and A.~M. Vijh (1993).
\newblock The valuation of path dependent contracts on the average.
\newblock {\em Management Science\/}~{\em 39\/}(10), 1202--1213.

\bibitem[\protect\citeauthoryear{Rogers and Shi}{Rogers and
  Shi}{1995}]{rogers1995value}
Rogers, L. C.~G. and Z.~Shi (1995).
\newblock The value of an {A}sian option.
\newblock {\em Journal of Applied Probability\/}~{\em 32\/}(4), 1077--1088.

\bibitem[\protect\citeauthoryear{Shaw}{Shaw}{2002}]{shaw2002pricing}
Shaw, W. (2002).
\newblock Pricing {A}sian options by contour integration, including asymptotic
  methods for low volatility.
\newblock {\em Working paper\/}.

\bibitem[\protect\citeauthoryear{Sun, Chen, and Li}{Sun
  et~al.}{2013}]{sun2013quasi}
Sun, J., L.~Chen, and S.~Li (2013).
\newblock A quasi-analytical pricing model for arithmetic {A}sian options.
\newblock {\em Journal of Futures Markets\/}~{\em 33\/}(12), 1143--1166.

\bibitem[\protect\citeauthoryear{Thompson}{Thompson}{2002}]{thompson2002fast}
Thompson, G. (2002).
\newblock Fast narrow bounds on the value of {A}sian options.
\newblock Technical report, Judge Institute of Management Studies.

\bibitem[\protect\citeauthoryear{Turnbull and Wakeman}{Turnbull and
  Wakeman}{1991}]{turnbull1991quick}
Turnbull, S.~M. and L.~M. Wakeman (1991).
\newblock A quick algorithm for pricing {E}uropean average options.
\newblock {\em Journal of Financial and Quantitative Analysis\/}~{\em 26\/}(3),
  377--389.

\bibitem[\protect\citeauthoryear{Turner}{Turner}{1966}]{turner1966inverse}
Turner, L.~R. (1966).
\newblock Inverse of the {V}andermonde matrix with applications.

\bibitem[\protect\citeauthoryear{Vanmaele, Deelstra, Liinev, Dhaene, and
  Goovaerts}{Vanmaele et~al.}{2006}]{vanmaele2006bounds}
Vanmaele, M., G.~Deelstra, J.~Liinev, J.~Dhaene, and M.~J. Goovaerts (2006).
\newblock Bounds for the price of discrete arithmetic {A}sian options.
\newblock {\em Journal of Computational and Applied Mathematics\/}~{\em
  185\/}(1), 51--90.

\bibitem[\protect\citeauthoryear{Vecer}{Vecer}{2001}]{vecer2001new}
Vecer, J. (2001).
\newblock A new {PDE} approach for pricing arithmetic average {A}sian options.
\newblock {\em Journal of Computational Finance\/}~{\em 4\/}(4), 105--113.

\bibitem[\protect\citeauthoryear{Vecer}{Vecer}{2002}]{vecer2002unified}
Vecer, J. (2002).
\newblock Unified pricing of {A}sian options.
\newblock {\em Risk\/}~{\em 15\/}(6), 113--116.

\bibitem[\protect\citeauthoryear{Yor}{Yor}{1992}]{yor1992some}
Yor, M. (1992).
\newblock On some exponential functionals of {B}rownian motion.
\newblock {\em Advances in Applied Probability\/}~{\em 24\/}(3), 509--531.

\bibitem[\protect\citeauthoryear{Zvan, Forsyth, and Vetzal}{Zvan
  et~al.}{1996}]{zvan1996robust}
Zvan, R., P.~A. Forsyth, and K.~R. Vetzal (1996).
\newblock Robust numerical methods for {PDE} models of {A}sian options.
\newblock Technical report, University of Waterloo, Faculty of Mathematics.

\end{thebibliography}
\bibliographystyle{chicago}

\clearpage
%%%%%%%%%%%%%%%%%%%%%%%%%%%%%%%%%%%%%%%%%%%%%%%%%%%%%%%%%%
%% Appendix
%%%%%%%%%%%%%%%%%%%%%%%%%%%%%%%%%%%%%%%%%%%%%%%%%%%%%%%%%%

%\appendixtitleon
\appendixtitletocon
%\appendix
\begin{appendices}
\section{Scaling with auxiliary moments}\label{appendix:scaling}
In this appendix we show how to avoid rounding errors by scaling the problem using the moments of the auxiliary density $w$.

Using $(L^{-1})^\top L^{-1}=M^{-1}$ we get from \eqref{eq:truncated_series}:
\begin{align*}
\pi^{(N)}&=(f_0,\ldots,f_N)(\ell_0,\ldots,\ell_N)^\top\\
&=\e^{-rT}(\tilde{f}_0,\ldots,\tilde{f}_N)M^{-1}\e^{G_N T}H_N(0).
\end{align*}
Define $S\in\R^{(N+1)\times(N+1)}$ as the diagonal matrix with the moments of $w$ on its diagonal:
\begin{align*}
S=\text{diag}(s_0,\ldots,s_N),\quad s_i=\e^{i\mu+\frac{1}{2}i^2\nu^2}.
\end{align*}
We can now write
\begin{align}
\pi^{(N)}&=\e^{-rT}(\overline{f}_0,\ldots,\overline{f}_N)S^{-1} S M^{-1}SS^{-1}\e^{G_N T}SS^{-1}H_N(0) \nonumber \\
&=\e^{-rT}(\overline{f}_0,\ldots,\overline{f}_N)\overline{M}^{-1}\e^{\overline{G}_N T}H_N(0),\label{eq:option_price_scaled}
\end{align}
where we have defined the matrices $((\overline{G}_N)_{ij})_{0\le i,j\le N}$, $(\overline{M}_{ij})_{0\le i,j\le N}$ and the vector $(\overline{f}_0,\ldots,\overline{f}_N)^\top$ as
\begin{align*}
\overline{M}_{ij}=\e^{ij\nu^2},\quad 
\overline{f}_i=\e^{\mu+\frac{1}{2}(2i+1)\nu^2}\Phi(d_{i+1})-K\Phi(d_i),\quad 
(\overline{G}_N)_{ij}=
\begin{cases}
ir+\frac{1}{2} i(i-1)\sigma^2& j=i\\
\frac{i}{T}\e^{-\mu+\frac{1}{2}(1-2i)\nu^2}& j=i-1\\
0& else
\end{cases},
\end{align*}
for $i,j=0,\ldots,N$.
The components of $\overline{M}$ and $(\overline{f}_0,\ldots,\overline{f}_N)^\top$ grow \emph{much} slower for increasing $N$ as their counterparts $M$ and $(\tilde{f}_0,\ldots,\tilde{f}_N)^\top$, respectively. The vector $\e^{\overline{G}_N T}H_N(0)$ corresponds to the moments of $A_T$ divided by the moments corresponding to $w$. Since both moments grow approximately at the same rate, this vector will have components around one. This scaling is important since for large $N$ the raw moments of $A_T$ are enormous. This was causing trouble for example in \cite{dufresne2000laguerre}. We therefore circumvent the numerical inaccuracies coming from the explosive moment behavior by directly computing the relative moments.

The likelihood and payoff coefficients can be computed by performing a Cholesky decomposition on $\overline{M}$ instead of $M$:
\begin{align*}
(f_0,\ldots,f_N)^\top&=\e^{-rT} \overline{L}^{-1}(\overline{f}_0,\ldots,\overline{f}_N)^\top,\\
(\ell_0,\ldots,\ell_N)^\top&=\overline{L}^{-1}\e^{\overline{G}_N T}H_N(0),
\end{align*}
where $\overline{M}=\overline{L}\,\overline{L}^\top$ is the Cholesky decomposition of $\overline{M}$.

Remark that to compute the option price in \eqref{eq:option_price_scaled}, we do not necessarily need to do a Cholesky decomposition. Indeed, we only need to invert the matrix $\overline{M}$. Doing a Cholesky decomposition is one way to solve a linear system, but there are other possibilities. In particular, remark that $\overline{M}$ is a Vandermonde matrix and its inverse can be computed analytically (see e.g. \cite{turner1966inverse}). There also exist specific numerical methods to solve linear Vandermonde systems, see e.g. \cite{bjorck1970solution}. However, we have not found any significant differences between the Cholesky method and alternative matrix inversion techniques for the examples considered in this paper. 

For very large values of $\nu$, even the matrix $\overline{M}$ might become ill conditioned. In this case it is advisable to construct the orthonormal basis using the three-term recurrence relation:
\begin{lemma} \label{lemma:orthonormal_polynomials}
The polynomials $b_n\in \mathrm{Pol}_n(\R)$, $n=0,1,\ldots$, defined recursively by
 \begin{align*}
&b_0(x)=1,\quad b_1(x)=\frac{1}{\beta_{1}} (x-\alpha_{0}),\quad \\
& b_{n}(x)=\frac{1}{\beta_{n}}\left( (x-\alpha_{n-1})b_{n-1}(x)-\beta_{n-1}b_{n-2}(x)\right),\quad n=2,3,\ldots,
\end{align*}
with
\begin{align*}
&\alpha_n=\e^{\mu+\nu^2(n-\frac{1}{2})}\left(\e^{\nu^2(n+1)}+\e^{\nu^2 n}-1\right),\quad n=0,1,\ldots \\
&\beta_{n}=\e^{\mu+\frac{1}{2}\nu^2(3n-2)}\sqrt{\e^{\nu^2 n}-1},\quad n=1,2,\ldots,
\end{align*}
satisfy
\begin{align*}
\int_\R b_i(x)b_j(x)w(x)\,\dd x=\begin{cases}1& i=j\\0& else\end{cases},\quad i,j=0,1,\ldots.
\end{align*}
\end{lemma}
\begin{proof}
Straightforward application of the moment-generating function of the normal distribution and Theorem 1.29 in \cite{gautschi2004orthogonal}.
\end{proof}
The above recursion suffers from rounding errors in double precision arithmetic for small $\nu$, but is very accurate for large $\nu$.

\section{Control variate for simulating $g(x)$} \label{appendix:control_variate}
In order to increase the efficiency of the Monte-Carlo simulation of $g(x)$, we describe in this appendix how to use the density of the geometric average as a control variate. This idea is inspired by \cite{kemna1990pricing}, who report a very substantial variance reduction when using the geometric Asian option price as a control variate in the simulation of the arithmetic Asian option price. Denote by $Q_T=\exp\left(\frac{1}{T}\int_0^T \log(S_s)\,\dd s\right)$ the geometrical price average. It is not difficult to see that $\log(Q_T)$ is normally distributed with mean $\frac{1}{2}(r-\frac{1}{2}\sigma^2)T$ and variance $\frac{\sigma^2}{3}T$. Hence, $Q_T$ is log-normally distributed and we know its density function, which we denote by $q(x)$, explicitly. Similarly as in Lemma \ref{lemma:malliavin}, we can also express $q(x)$ as an expectation:
\begin{lemma}
For any $x\in\R$
\begin{align}
q(x)=\E\left[ \left(1_{\{Q_T\ge x\}} -c(x)\right)\left( \frac{2B_T}{\sigma T Q_T}+\frac{1}{Q_T}\right)\right],\label{eq:density_malliavin_geometric}
\end{align}
where $c$ is any deterministic finite-valued function.
\label{lemma:malliavin_geometric}
\end{lemma}

We now propose the following estimator for the density of the arithmetic average:
\begin{align}
g(x)=&\E\left[ (1_{\{A_T\ge x\}} -c_1(x)) \frac{2}{\sigma^2}\left( \frac{S_T-S_0}{T A_T^2}+\frac{\sigma^2-r}{A_T}\right)\right. \nonumber\\
&\left.+\left(q(x)-\left(1_{\{Q_T\ge x\}} -c_2(x)\right)\left( \frac{2B_T}{\sigma T Q_T}+\frac{1}{Q_T}\right)\right)\right],
\label{eq:density_control_variate}
\end{align}
for some deterministic finite-valued functions $c_1$ and $c_2$. Given the typically high correlation between the geometric and arithmetic average, the above estimator has a significantly smaller variance than the estimator in \eqref{eq:density_malliavin}. In numerical examples the functions $c_1$ and $c_2$ are chosen as follows:
\begin{align*}
c_1(x)=1_{x\le m^A_1},\quad c_2(x)=1_{x\le m^Q_1},
\end{align*}
where $m^A_1$ and $m^Q_1$ denote the first moments of $A_T$ and $Q_T$, respectively.

Finally, we use \eqref{eq:density_control_variate} to express $\lVert \ell\rVert^2_w$ as an expectation that can be evaluated using Monte-Carlo simulation:
\begin{align}
\lVert \ell\rVert^2_w=
\E&\left[ \frac{(1_{\{A_T\ge \tilde{A}_T\}} -c_1(\tilde{A}_T)) \frac{2}{\sigma^2}\left( \frac{S_T-S_0}{T A_T^2}+\frac{\sigma^2-r}{A_T}\right)}{w(\tilde{A}_T)} \right.\nonumber\\
&+\left. \frac{q(\tilde{A}_T)-\left(1_{\{Q_T\ge \tilde{A}_T\}} -c_2(\tilde{A}_T)\right)\left( \frac{2B_T}{\sigma T Q_T}+\frac{1}{Q_T}\right)}{w(\tilde{A}_T)}\right],
\label{eq:likelihood_ratio_MC_control_variate}
\end{align}
where the random variable $\tilde{A}_T$ is independent from all other random variables and has the same distribution as $A_T$.
In numerical examples we find a variance reduction of roughly a factor ten.

\section{Proofs}\label{appendix:proofs}
This appendix contains all the proofs.
\subsection{Proof of Lemma \ref{lemma:identity_law}}
Using the time-reversal property of a Brownian motion, we have the following identity in law for fixed $t>0$ :
\begin{align*}
t A_t&=\int^t_0 \e^{(r-\frac{1}{2}\sigma^2)u+\sigma B_u}\,\dd u\\
&\overset{\text{law}}{=}
\int^t_0 \e^{(r-\frac{1}{2}\sigma^2)(t-u)+\sigma (B_t-B_u)}\,\dd u\\
&=S_t \int_0^t S_u^{-1}\,\dd u.
\end{align*}
Applying It\^o's lemma to $X_t:=S_t \int_0^t S_u^{-1}\,\dd u$ gives
\begin{align*}
\dd X_t&=S_t S_t^{-1}\,\dd t+  \int_0^t S_u^{-1}\,\dd u \,(rS_t\,\dd t+\sigma S_t\,\dd B_t)\\
&=(rX_t +1)\,\dd t+\sigma X_t\,\dd B_t.
\end{align*}

\subsection{Proof of Proposition \ref{prop:moments}}
Applying the infinitesimal generator $\Gcal$ corresponding to the diffusion in \eqref{eq:garch_diffusion} to a monomial $x^n$ gives:
\begin{align*}
\Gcal x^n=x^n(nr+\frac{1}{2} n(n-1)\sigma^2)+nx^{n-1}.
\end{align*}
Hence, we have $\Gcal H_n(x)=\tilde{G}_nH_n(x)$ componentwise, where $\tilde{G}_n$ is defined as
\begin{align*}
\tilde{G}_n=
\begin{pmatrix}
0&\\
1&r&\\
&\ddots&\ddots&\\
&&n&(nr+\frac{1}{2}n(n-1)\sigma^2)
\end{pmatrix}.
\end{align*}
 
Using the identity in distribution of Lemma \ref{lemma:identity_law}, we get
\begin{align*}
\E[H_n(A_T)]&=\mathrm{diag}(H_n(T^{-1}))\E[H_n(X_T)]\\
&=\mathrm{diag}(H_n(T^{-1}))\left(H_n(X_0)+\int_0^T\E[\Gcal H_n(X_u)]\,\dd u\right) \\
&=\mathrm{diag}(H_n(T^{-1}))H_n(0)+\mathrm{diag}(H_n(T^{-1}))\tilde{G}_n\int_0^T\E[H_n(X_u)]\,\dd u\\
&=H_n(0)+\mathrm{diag}(H_n(T^{-1}))\tilde{G}_n\mathrm{diag}(H_n(T))\int_0^T\,\E[H_n(A_u)]\,\dd u\\
&=H_n(0)+G_n\int_0^T\,\E[H_n(A_u)]\,\dd u,
\end{align*}
where $G_n$ was defined in \eqref{eq:generator_matrix}. The result now follows from solving the above linear first order ODE.

\subsection{Proof of Proposition \ref{prop:upper_tail}}
\begin{enumerate}
\item We will show that the solution at time $T>0$ of the SDE in \eqref{eq:garch_diffusion} admits a smooth density function. The claim then follows by the identity in law.

Define the volatility and drift functions $a(x)=\sigma x$ and $b(x)=rx+1$. H\"ormander's condition (see for example Section 2.3.2 in \cite{nualart2006malliavin}) becomes in this case:
\begin{align*}
a(X_0)\neq 0 \quad \text{or}\quad a^{(n)}(X_0)b(X_0)\neq 0 \,\, \text{for some $n\ge 1$}.
\end{align*}
H\"ormander's condition is satisfied since for $n=1$ we have $a'(0)b(0)=\sigma \neq 0$. Since $a(x)$ and $b(x)$ are infinitely differentiable functions with bounded partial derivatives of all orders, we conclude by Theorem 2.3.3 in \cite{nualart2006malliavin} that $X_T$, and therefore $A_T$, admits a smooth density function.

\item 

We start from the following two observations:
\[
A_T\le \sup_{0\le u\le T} S_u\quad  \text{and}\quad
P\left(\sup_{0\le u\le T} \sigma B_u \ge x\right)=2P\left(Z\ge \frac{x}{\sigma \sqrt{T}}\right),
\]
where $Z\sim N(0,1)$. Using the fact that the exponential is an increasing function, we get
\begin{align*}
P(A_T\ge x)&\le P\left( \sup_{0\le u\le T} S_u\ge x\right)\\
&=P\left( \sup_{0\le u\le T} (r-\frac{1}{2}\sigma^2)u+\sigma B_u \ge \log(x)\right)\\
&\le 
\begin{cases}
2P\left(Z\ge \frac{\log(x)-(r-\frac{1}{2}\sigma^2) T}{\sigma\sqrt{T}}\right)& \text{if}\quad r\ge \frac{1}{2}\sigma^2\\
2P\left(Z\ge \frac{\log(x)}{\sigma\sqrt{T}}\right)& \text{if}\quad r\le \frac{1}{2}\sigma^2
\end{cases}.
\end{align*}
Applying the rule of l'H\^opital gives
\begin{align*}
\lim_{x\to\infty} g(x)\left(\e^{-\frac{\log(x)^2}{2\sigma^2 T}}\right)^{-1}&
=\lim_{x\to\infty}P(A_T\ge x)\left(\int_x^\infty\e^{-\frac{\log(y)^2}{2\sigma^2 T}}\,\dd y\right)^{-1}\\
&\le 2\frac{1}{\sqrt{2\pi T}\sigma}\lim_{x\to\infty}\int_x^\infty\e^{-\frac{(\log(y)-(r-\frac{1}{2}\sigma^2)^+T)^2}{2\sigma^2 T}}\,\dd y
\left(\int_x^\infty\e^{-\frac{\log(y)^2}{2\sigma^2 T}}\,\dd y\right)^{-1}\\
&= \sqrt{\frac{2}{\pi T}}\frac{1}{\sigma}.
\end{align*}
Hence we have show that $g(x)=\Ocal\left(\exp\left\{-\frac{1}{2}\frac{\log(x)^2}{\sigma^2 T}\right\} \right)$ for $x\to\infty$.

Since the exponential is a convex function, we have that the arithmetic average is always bounded below by the geometric average:
\begin{align*}
A_T\ge Q_T=\exp\left(\frac{1}{T}\int_0^T \log(S_s)\,\dd s\right).
\end{align*}
It is not difficult to see that $\log(Q_T)$ is normally distributed with mean $\frac{1}{2}(r-\frac{1}{2}\sigma^2)T$ and variance $\frac{\sigma^2}{3}T$. By similar arguments as before we therefore have
\begin{align*}
g(x)=\Ocal \left( \exp\left\{-\frac{3}{2}\frac{\log(x)^2}{\sigma^2T}\right\} \right)\quad \text{for}\quad x\to 0.
\end{align*}

\end{enumerate}

\subsection{Proof of Proposition \ref{prop:likelihood_ratio}}
We can write the squared norm of $\ell$ as
\begin{align}
\lVert \ell\rVert_w^2&=\int_{0}^\infty \left(\frac{g(x)}{w(x)}\right)^2 w(x)\,\dd x\nonumber\\
&=
\int_0^a\frac{g(x)^2}{w(x)}\,\dd x+
\int_a^b \frac{g(x)^2}{w(x)}\,\dd x+
\int_b^\infty \frac{g(x)^2}{w(x)}\,\dd x,\label{eq:proof:lr}
\end{align}
for some $0<a<b<\infty$. The second term is finite since the function $\frac{g^2}{w}$ is continuous over the compact interval $[a,b]$. From Proposition \ref{prop:upper_tail} we have
\[
g(x)^2=\Ocal\left(\exp\left\{-\frac{\log(x)^2}{\sigma^2T}\right\} \right),\quad \text{for}\quad x\to\infty \quad \text{and}\quad x\to 0.
\] 
For the log-normal density we have
\[
w(x)=\Ocal\left(\exp\left\{-\frac{\log(x)^2}{2\nu}\right\} \right),\quad \text{for}\quad x\to\infty \quad \text{and}\quad x\to 0.
\] 
Since $2\nu> \sigma^2T$ by assumption, we are guaranteed that the first and last term in \eqref{eq:proof:lr} are finite for a sufficiently small (resp.\ large) choice of $a$ (resp.\ $b$).

\subsection{Proof of Proposition \ref{prop:payoff_coeff}}
The payoff coefficients can be written as
\begin{align*}
(f_0,\ldots,f_N)^\top&=e^{-rT}C(\tilde{f}_0,\ldots,\tilde{f}_N)^\top,
\end{align*}
with
\begin{align*}
\tilde{f}_n&=\frac{1}{\sqrt{2\pi}\nu}\int_0^\infty(e^x-K)^+\e^{nx}\e^{-\frac{(x-\mu)^2}{2\nu^2}}\,\dd x\\
&=\frac{1}{\sqrt{2\pi}\nu}\left(\int_{\log(K)}^\infty \e^{(n+1)x}\e^{-\frac{(x-\mu)^2}{2\nu^2}}\,\dd x-K\int_{\log(K)}^\infty \e^{nx}\e^{-\frac{(x-\mu)^2}{2\nu^2}}\,\dd x\right).
\end{align*}
Completing the square in the exponent gives
\begin{align*}
\frac{1}{\sqrt{2\pi}\nu}\int_{\log(K)}^\infty \e^{nx}\e^{-\frac{(x-\mu)^2}{2\nu^2}}\,\dd x
&=\frac{1}{\sqrt{2\pi}\nu}\e^{\mu n +\frac{1}{2}n^2\nu^2}\int_{\log(K)}^\infty\e^{-\frac{(x-(\mu+\nu^2 n))^2}{2\nu^2}}\,\dd x\\
&=\frac{1}{\sqrt{2\pi}}\e^{\mu n +\frac{1}{2}n^2\nu^2}\int_{\frac{\log(K)-(\mu+\nu^2 n)}{\nu}}^\infty\e^{-\frac{1}{2}y^2}\,\dd y\\
&=\e^{\mu n +\frac{1}{2}n^2\nu^2}\Phi(d_n),
\end{align*}
where $d_n$ is defined in \eqref{eq:payoff_coeff_d}. We finally get
\begin{align*}
\tilde{f}_n=\e^{\mu (n+1) +\frac{1}{2}(n+1)^2\nu^2}\Phi(d_{n+1})-K\e^{\mu n +\frac{1}{2}n^2\nu^2}\Phi(d_n).
\end{align*}

\subsection{Proof of Lemma \ref{lemma:malliavin}}
This proof is based on Malliavin calculus techniques, we refer to \cite{nualart2006malliavin} for an overview of standard results in this area. A similar approach is taken by \cite{fournie1999applications} to compute the Greeks of an Asian option by Monte-Carlo simulation.

Denote by $D: \mathbb{D}^{1,2}\to L^2(\Omega\times [0,T]),\, F\mapsto \left\{D_t F, t\in [0,T]\right\}$, the Malliavin derivative operator. By Theorem 2.2.1 in \citep{nualart2006malliavin} we have $S_t,A_t\in \mathbb{D}^{1,2}$ for $t\in (0,T]$ and
\begin{align*}
D_uS_t=\sigma S_t 1_{\left\{u\le t\right\}},\quad D_u A_t=\frac{\sigma}{t} \int_u^tS_s\,\dd s.
\end{align*}
Denote by 
\[\delta: Dom(\delta)\to L^2(\Omega),\, \left\{ X_t, t\in[0,T]\right\}\mapsto \delta(X)\] 
the Skorohod integral and by $Dom(\delta)\subseteq  L^2(\Omega\times [0,T])$ the corresponding domain. The Skorohod integral is defined as the adjoint operator of the Malliavin derivative and can be shown to extend the It\^o integral to non-adapted processes. In particular, we have immediately that $\left\{ S_t, t\in[0,T]\right\}\in Dom(\delta)$ and 
\begin{equation}
\delta(S)=\int_0^T S_s\,\dd B_s.\label{eq:ito_skor}
\end{equation}

For $\phi\in C^\infty_c$ we have $\phi(A_T)\in \mathbb{D}^{1,2}$ and 
\begin{align*}
\int_0^T (D_u\phi(A_T))S_u\,\dd u=\phi'(A_T)\int_0^T (D_u A_T) S_u\,\dd u.
\end{align*}
Using the duality relationship between the Skorohod integral and the Malliavin derivative we get
\begin{align}
\E[\phi'(A_T)]&=\E\left[\int_0^T (D_u\phi(A_T))\frac{S_u}{\int_0^T (D_u A_T) S_u\,\dd u}\,\dd u\right]\nonumber\\
&=\E\left[\phi(A_T)\delta\left(\frac{S}{\int_0^T (D_u A_T) S_u\,\dd u}\right)\right].\label{eq:IBP_malliavin}
\end{align}
%If we denote by $\delta_0$ the Dirac function, then we can formally write
%\begin{align*}
%g(x)&=\E[\delta_0(A_T-x)]=\E\left[\partial_y 1_{[x,\infty)}(A_T)\right].
%\end{align*}
By Lemma 1 in \cite{bally2003elementary} (see also Proposition 2.1.1 in \cite{nualart2006malliavin} for a similar approach) we obtain the following representation of the density function of $A_T$:\footnote{Informally speaking one applies a regularization argument in order to use \eqref{eq:IBP_malliavin} for the (shifted) Heaviside function $\phi(y)=1_{\{y\ge x\}}$.}
\begin{align}
g(x)&=\E\left[1_{\{A_T\ge x\}}\delta\left(\frac{S}{\int_0^T (D_u A_T) S_u\,\dd u}\right)\right] \nonumber\\
&=\frac{T}{\sigma} \E\left[1_{\{A_T\ge x\}}\delta\left(\frac{S}{\int_0^T S_u \int_u^T S_s\,\dd s \,\dd u}\right)\right].\label{eq:density_malliavin_proof}
\end{align}
Interchanging the order of integration gives
\begin{align*}
\int_0^T S_u \int_u^T S_s\,\dd s\,\dd u&=\left(\int_0^T S_u\,\dd u\right)^2-\int_0^T  \int_0^u S_u S_s\,\dd s\,\dd u\\
&=\left(\int_0^T S_u\,\dd u\right)^2-\int_0^T S_s \int_s^T S_u\,\dd u\,\dd s,
\end{align*}
which gives $
\int_0^T S_u \int_u^T S_s\,\dd s\,\dd u=\frac{T^2}{2}A_T^2.
$
Plugging this into \eqref{eq:density_malliavin_proof} gives
\begin{align*}
g(x)&=\frac{2}{T\sigma} \E\left[1_{\{A_T\ge x\}} \delta\left(\frac{S}{ A_T^{2}}\right)\right].
\end{align*}
We use Proposition 1.3.3 in \cite{nualart2006malliavin} to factor out the random variable $A_T^{-2}$ from the Skorohod integral:
\begin{align*}
\delta\left(\frac{S}{ A_T^{2}}\right)&=
A_T^{-2}\delta(S)-\int_0^TD_t\left(A_T^{-2}\right) S_t\,\dd t\\
&=
A_T^{-2}\frac{1}{\sigma}\left(S_T-S_0-r\int_0^TS_s\,\dd s\right)-\int_0^TD_t\left(A_T^{-2}\right) S_t\,\dd t,
\end{align*}
where we used \eqref{eq:ito_skor} in the last equation. Using the chain rule for the Malliavin derivative we get
\begin{align*}
\delta\left(\frac{S}{ A_T^{2}}\right)&=
A_T^{-2}\frac{1}{\sigma}\left(S_T-S_0-r\int_0^TS_s\,\dd s\right)+2A_T^{-3}\int_0^TD_t A_T S_t\,\dd t\\
&=
A_T^{-2}\frac{1}{\sigma}\left(S_T-S_0-rT A_T\right)+2A_T^{-3}\frac{1}{T}\int_0^T S_t \int_t^T\sigma S_u\,\dd u \,\dd t\\
&=
A_T^{-2}\frac{1}{\sigma}\left(S_T-S_0-r T A_T\right)+ A_T^{-1}\sigma T\\
&=A_T^{-2}\frac{1}{\sigma}\left(S_T-S_0\right)+\frac{T}{\sigma} A_T^{-1}(\sigma^2-r).
\end{align*}
Putting everything back together we finally get:
\begin{align*}
g(x)=\frac{2}{\sigma^2}\E\left[ 1_{\{A_T\ge x\}} \left( \frac{S_T-S_0}{T A_T^2}+\frac{\sigma^2-r}{A_T}\right)\right].
\end{align*}
Since the Skorohod integral has zero expectation we also have
\begin{align*}
g(x)=\frac{2}{\sigma^2}\E\left[ \left( 1_{\{A_T\ge x\}}-c(x)\right)\left( \frac{S_T-S_0}{T A_T^2}+\frac{\sigma^2-r}{A_T}\right)\right],
\end{align*}
for any deterministic finite-valued function $c$.

\subsection{Proof of Corollary \ref{corollary:likelihood_ratio_MC}}
The result follows immediately from \eqref{eq:density_malliavin} and
\begin{align*}
\lVert \ell \rVert^2_w&= \int_0^\infty \ell^2(x)w(x)\,\dd x=\int_0^\infty \frac{g(x)}{w(x)}g(x)\,\dd x .
\end{align*}

\subsection{Proof of Proposition \ref{prop:error_bounds} }
Using the Cauchy-Schwarz inequality and the orthonormality of the polynomials $b_0,\ldots,b_N$ we get 
\begin{align*}
|\pi-\pi^{(N)}|&=\left|\langle F,l\rangle_w-\sum_{n=0}^N f_n \ell_n\right|\\
&=\left\langle F- \sum_{n=0}^Nb_n f_n \, ,\, \ell- \sum_{n=0}^Nb_n \ell_n\right\rangle_w\\
&\le \left\lVert F- \sum_{n=0}^Nb_n f_n \right \rVert_w \, \left\lVert\ell- \sum_{n=0}^Nb_n \ell_n \right \rVert_w\\
&=\left(\lVert F\rVert^2_w- \sum_{n=0}^N f_n^2\right)^{\frac{1}{2}}\left(\lVert \ell\rVert^2_w- \sum_{n=0}^N \ell_n^2\right)^{\frac{1}{2}}.
\end{align*}

\subsection{Proof of Lemma \ref{lemma:malliavin_geometric}}
Applying the Malliavin derivative to $Q_T$ gives
\begin{align*}
D_uQ_T=Q_T D_u \left(\frac{1}{T}\int_0^T \log(S_s)\,\dd s\right)=Q_T\frac{\sigma}{T}(T-u)1_{u\le T}.
\end{align*}
Similarly as in the proof of Lemma \ref{lemma:malliavin} we can write
\begin{align}
q(x)&=\E\left[1_{\{Q_T\ge x\}}\delta\left(\frac{1}{\int_0^T D_u Q_T\,\dd u}\right)\right]\nonumber \\
&= \E\left[1_{\{Q_T\ge x\}}\delta\left(\frac{2}{Q_T\sigma T}\right)\right].\label{eq:density_malliavin_geometric_proof}
\end{align}
Using Proposition 1.3.3 in \cite{nualart2006malliavin} to factor out the random variable from the Skorohod integral gives
\begin{align*}
 \delta\left(\frac{2}{Q_T\sigma T}\right)&=\frac{2B_T}{\sigma TQ_T}-\frac{2}{\sigma T}\int_0^T D_u(Q_T^{-1})\,\dd u\\
 &=\frac{2B_T}{\sigma TQ_T}-\frac{2}{\sigma T}Q_T^{-2}\int_0^T Q_T\frac{\sigma}{T}(T-u)\,\dd u\\
 &=\frac{2B_T}{\sigma T Q_T}+\frac{1}{Q_T}.
\end{align*}
Plugging this back into \eqref{eq:density_malliavin_geometric_proof} finally gives
\begin{align*}
q(x)=\E\left[1_{\{Q_T\ge x\}}\left(\frac{2B_T}{\sigma T Q_T}+\frac{1}{Q_T}\right)\right].
\end{align*}
Since the Skorohod integral has zero expectation we also have
\begin{align*}
q(x)=\E\left[\left(1_{\{Q_T\ge x\}}-c(x)\right)\left(\frac{2B_T}{\sigma T Q_T}+\frac{1}{Q_T}\right)\right],
\end{align*}
for any deterministic finite-valued function $c$.
\end{appendices}

\clearpage
%%%%%%%%%%%%%%%%%%%%%%%%%%%%%%%%%%%%%%%%%%%%%%%%%%%%%%%%%%
%% Tables
%%%%%%%%%%%%%%%%%%%%%%%%%%%%%%%%%%%%%%%%%%%%%%%%%%%%%%%%%%
\setlength{\tabcolsep}{3pt}
\begin{table}
\centering
\begin{tabular}{cccccccccccl@{\hspace{1.2pt}}c@{\hspace{1.2pt}}r}
\toprule
Case	& $r$	& $\sigma$	&$T$	&$S_0$	& LNS10	& LNS15&	LNS20 	&	LS	& EE	& VEC & 	\multicolumn{3}{c}{MC 95\% CI} \\\midrule
1		& .02	& .10		&1		& 2.0	&.05601	&.05600	&.05599 & .0197	& .05599& .05595& [.05598&,&.05599]	\\
2		&.18	&.30		&1		&2.0	&.2185 &.2184 	&.2184	&.2184	& .2184& .2184&	[.2183&,&.2185]	\\
3		&.0125	&.25		&2		&2.0	&.1723	&.1722 &.1722	&.1723	& .1723& .1723& [.1722&,&.1724]	\\
4		&.05	&.50		&1		&1.9	& .1930	&.1927	&.1928		&.1932	& .1932& .1932&[.1929&,& .1933]		\\
5		&.05	&.50		&1		&2.0	&.2466 &.2461 &.2461		&.2464	& .2464& .2464& [.2461&,&.2466]	\\
6		&.05	&.50		&1		&2.1	& .3068 &.3062 &.3061		&.3062	& .3062& .3062& [.3060&,&.3065] \\
7		&.05	&.50		&2		&2.0	&.3501&.3499 &.3499		&.3501	& .3501& .3500& [.3494&,& .3504]	\\
\bottomrule
\end{tabular}
\caption{Price approximations for different parameterizations and different methods. The strike price is $K=2$ for all cases. The column LNS$X$ refers to the method presented in this paper with the first $1+X$ terms of the series, LS to \cite{dufresne2000laguerre}, EE to \cite{linetsky2004spectral}, VEC to \cite{vecer2001new,vecer2002unified}, and MC 95\% CI to the 95\% confidence interval of the Monte-Carlo simulation.}
\label{table:benchmark_cases}
\end{table}

\begin{table}
\centering
\begin{tabular}{cccccccccccr}
\toprule
Case	& $r$	& $\sigma$	&$T$	&$S_0$	&	LNS10	& LNS15 & LNS20 &	LS & EE	& VEC&	\multicolumn{1}{c}{MC} \\\midrule
1 & .02 & .10 & 1 & 2.0 & .006 & .008 & .009 & .930 & - & .277 & 6.344 \\ 
2 & .18 & .30 & 1 & 2.0 & .002 & .002 & .003 & .666 & 2.901 & .345 & 5.518 \\ 
3 & .0125 & .25 & 2 & 2.0 & .002 & .002 & .002 & .635 & 3.505 & .374 & 12.138 \\ 
4 & .05 & .50 & 1 & 1.9 & .001 & .002 & .003 & .785 & 3.172 & .404 & 6.819 \\ 
5 & .05 & .50 & 1 & 2.0 & .001 & .002 & .002 & .701 & 2.768 & .404 & 5.432 \\ 
6 & .05 & .50 & 1 & 2.1 & .001 & .001 & .002 & .687 & 2.719 & .398 & 5.452 \\ 
7 & .05 & .50 & 2 & 2.0 & .002 & .002 & .004 & .594 & 2.202 & .438 & 11.699 \\  
\bottomrule
\end{tabular}
\caption{Computation times in seconds. The column LNS$X$ refers to the method presented in this paper with the first $1+X$ terms of the series, LS to \cite{dufresne2000laguerre}, EE to \cite{linetsky2004spectral}, VEC to \cite{vecer2001new,vecer2002unified}, and MC to the Monte-Carlo simulation.}
\label{table:computation_times}
\end{table}

\clearpage
%%%%%%%%%%%%%%%%%%%%%%%%%%%%%%%%%%%%%%%%%%%%%%%%%%%%%%%%%%
%% Figures
%%%%%%%%%%%%%%%%%%%%%%%%%%%%%%%%%%%%%%%%%%%%%%%%%%%%%%%%%%

\begin{figure}
\centering
\begin{subfigure}[b]{0.48\textwidth}
\includegraphics[width=\textwidth]{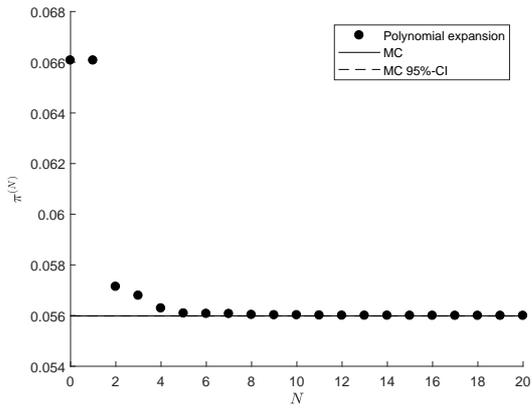}
\caption{Case 1}
\end{subfigure}
\begin{subfigure}[b]{0.48\textwidth}
\includegraphics[width=\textwidth]{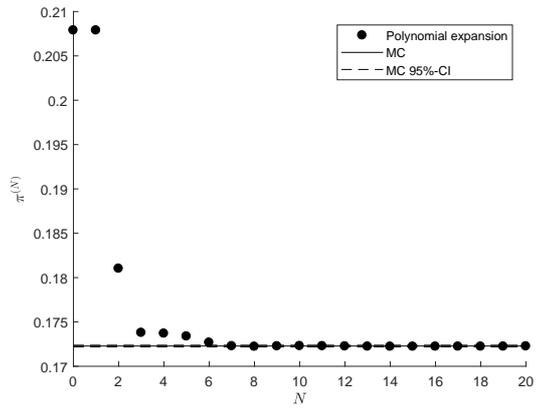}
\caption{Case 3}
\end{subfigure}
\begin{subfigure}[b]{0.48\textwidth}
\includegraphics[width=\textwidth]{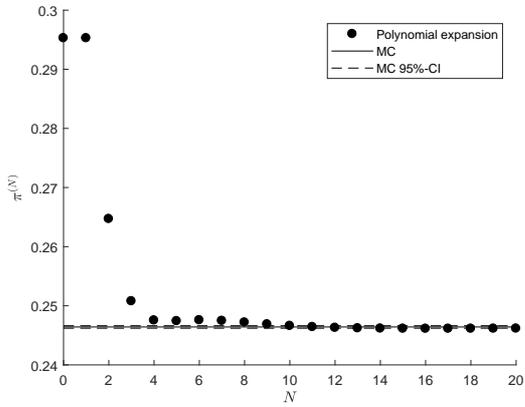}
\caption{Case 5}
\end{subfigure}
\begin{subfigure}[b]{0.48\textwidth}
\includegraphics[width=\textwidth]{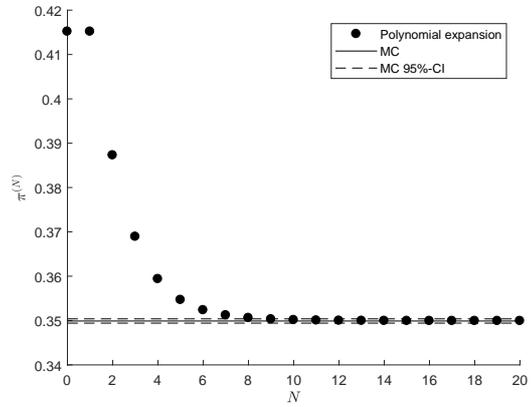}
\caption{Case 7}
\end{subfigure}
\caption{Asian option price approximations in function of polynomial approximation order $N$. The cases correspond to different parameterizations shown in Table \ref{table:benchmark_cases}.}
\label{fig:benchmark_cases}
\end{figure}

\begin{figure}
\centering
\begin{subfigure}[b]{0.48\textwidth}
\includegraphics[width=\textwidth]{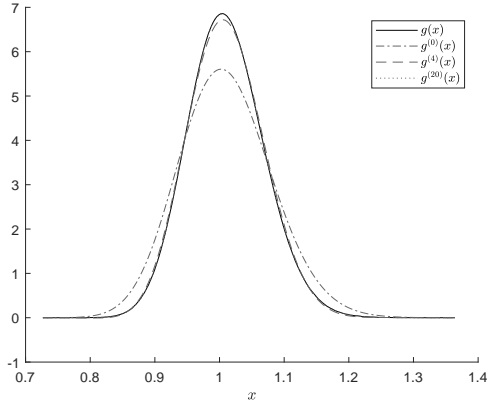}
\caption{Case 1}
\label{fig:density_approx_1}
\end{subfigure}
\begin{subfigure}[b]{0.48\textwidth}
\includegraphics[width=\textwidth]{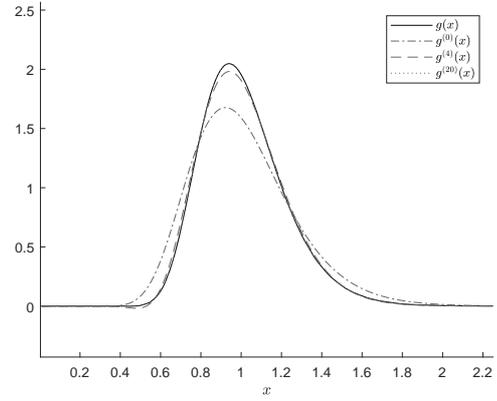}
\caption{Case 3}
\label{fig:density_approx_3}
\end{subfigure}
\begin{subfigure}[b]{0.48\textwidth}
\includegraphics[width=\textwidth]{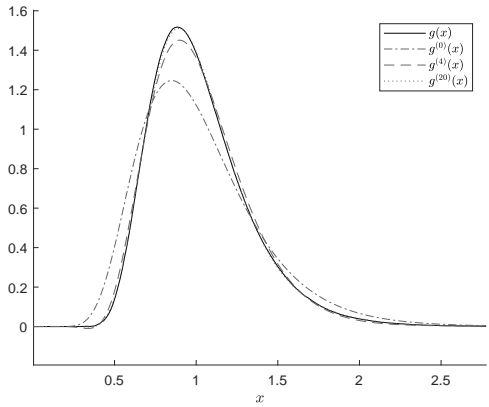}
\caption{Case 5}
\label{fig:density_approx_5}
\end{subfigure}
\begin{subfigure}[b]{0.48\textwidth}
\includegraphics[width=\textwidth]{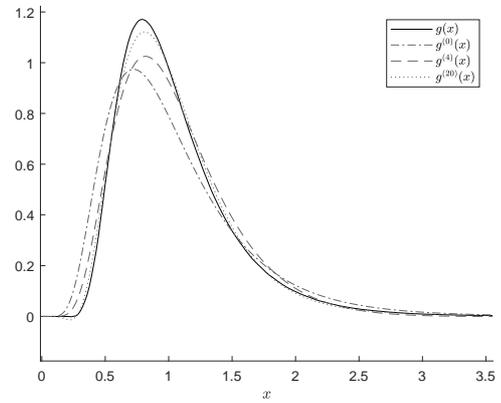}
\caption{Case 7}
\label{fig:density_approx_7}
\end{subfigure}
\caption{Simulated true density function $g(x)$ and approximated density functions $g^{(n)}(x)$, $n=0,4,20$. The cases correspond to different parameterizations shown in Table \ref{table:benchmark_cases}.}
\label{fig:density_approx}
\end{figure}

\begin{figure}
\centering
\begin{subfigure}[b]{0.48\textwidth}
\includegraphics[width=\textwidth]{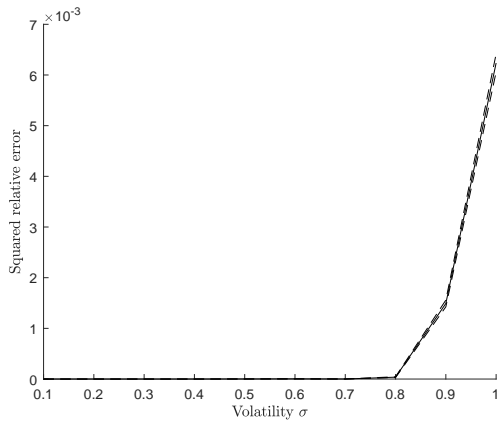}
\caption{Squared relative error}
\label{fig:range_vol_MC}
\end{subfigure}
\begin{subfigure}[b]{0.48\textwidth}
\includegraphics[width=\textwidth]{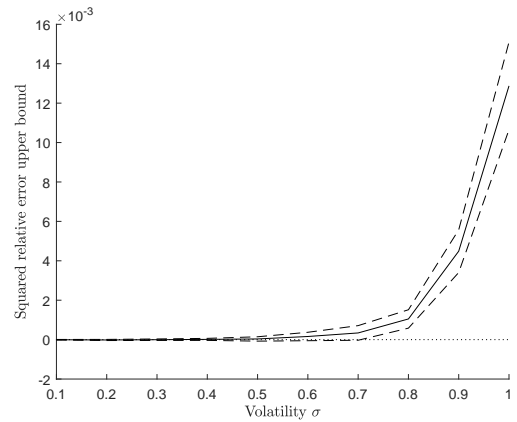}
\caption{Squared relative error upper bound}
\label{fig:range_vol_upper}
\end{subfigure}
\caption{Squared relative approximation error for different values of $\sigma$. Dashed lines correspond to the 95\% confidence intervals from the Monte-Carlo simulation. Parameters: $T=1$, $r=0.05$, $S_0=K=2$.}
\label{fig:range_vol}
\end{figure}

\begin{figure}
\centering
\begin{subfigure}[b]{0.48\textwidth}
\includegraphics[width=\textwidth]{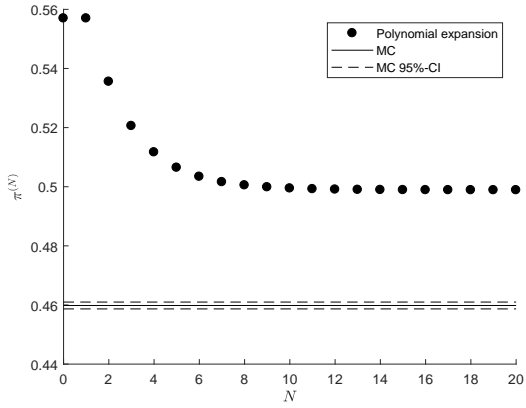}
\caption{Monte-Carlo simulated price}
\label{fig:high_vol_MC}
\end{subfigure}
\begin{subfigure}[b]{0.48\textwidth}
\includegraphics[width=\textwidth]{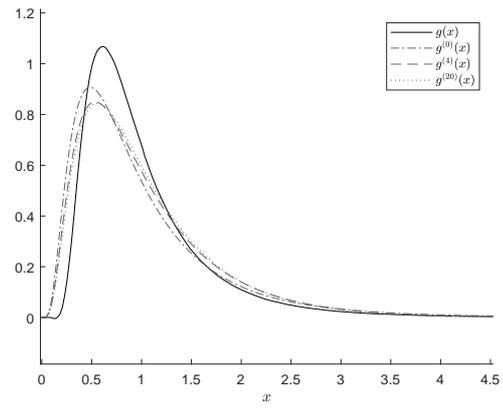}
\caption{Density approximations}
\label{fig:high_vol_density}
\end{subfigure}
\begin{subfigure}[b]{0.48\textwidth}
\includegraphics[width=\textwidth]{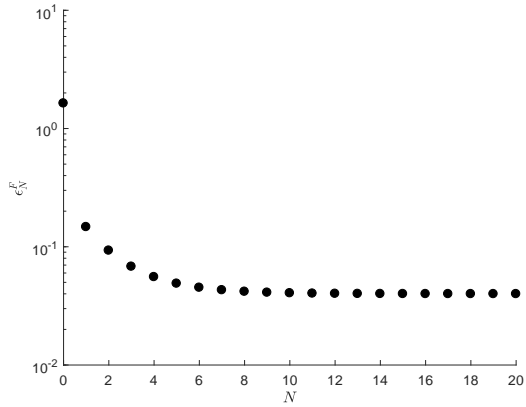}
\caption{Payoff projection error}
\label{fig:high_vol_payoffProjection}
\end{subfigure}
\begin{subfigure}[b]{0.48\textwidth}
\includegraphics[width=\textwidth]{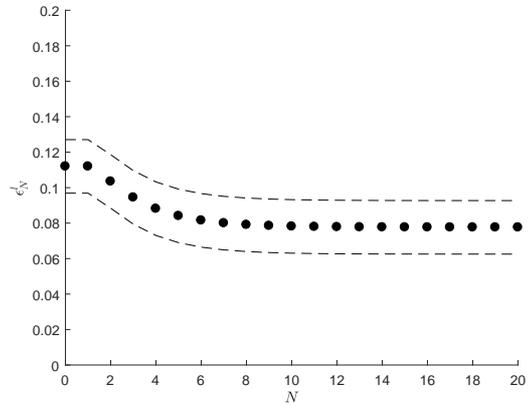}
\caption{Likelihood projection error}
\label{fig:high_vol_likelihoodProjection}
\end{subfigure}
\caption{Visualization of the projection bias in the extreme volatility case. Dashed lines correspond to the 95\% confidence intervals from the Monte-Carlo simulation. Parameters: $\sigma=1$, $T=1$, $r=0.05$, $S_0=K=2$.}
\label{fig:high_vol}
\end{figure}

\end{document}